\documentclass[runningheads]{llncs}

\usepackage{graphicx, float,hyperref}
\usepackage{qcircuit}
\usepackage{amsmath,mathtools,amsfonts,amsthm}
\usepackage[linesnumbered,ruled,vlined]{algorithm2e}
\usepackage{xy}
\usepackage[figurename=Circuit]{caption}

\DeclarePairedDelimiter\ket{\lvert}{\rangle}

\begin{document}
\authorrunning{Hoof, I van}
\titlerunning{Space-efficient quantum multiplication...}

\title{Space-efficient quantum multiplication of polynomials for binary finite fields with sub-quadratic Toffoli gate count}
\author{Iggy van Hoof
}

\institute{Technische Universiteit Eindhoven}

\maketitle

\begin{abstract}
    Multiplication is an essential step in a lot of calculations. In this paper we look at multiplication of 2 binary polynomials of degree at most $n-1$, modulo an irreducible polynomial of degree $n$ with $2n$ input and $n$ output qubits, without ancillary qubits, assuming no errors. With straightforward schoolbook methods this would result in a quadratic number of Toffoli gates and a linear number of CNOT gates. This paper introduces a new algorithm that uses the same space, but by utilizing space-efficient variants of Karatsuba multiplication methods it requires only $O(n^{\log_2(3)})$ Toffoli gates at the cost of a higher CNOT gate count: theoretically up to $O(n^2)$ but in examples the CNOT gate count looks a lot better.
\end{abstract}

\section{Introduction}
Multiplication of two polynomials in a finite field is an important step in many algorithms, such as point addition in elliptic curve cryptography. For classical computers a wealth of variations exist, often based around Karatsuba's multiplication method \cite{karatsuba1962multiplication}.

In the classical setting, temporary results for the steps of Karatsuba calculations have traditionally been stored separately. In 1993 Maeder \cite{maeder1993karatsuba} used around $2n$ additional space for multiplying degree-$n$ polynomials. This was improved by Thom\'e in 2002 to $n$ temporary space which at the time was believed to be optimal: ``it does not seem likely that anything better than this result can be obtained.'' \cite{thome2002karatsuba} However, in 2009 Roche did obtain a better result: $O(\log n)$ space Karatsuba multiplication of polynomials without additional time by doing many in-place operations \cite{roche2009space}. This was expanded by Cheng \cite{cheng16space} to also work for integers. Despite the advantages these variants offer, these methods are still relatively unknown.

This bound of $O(\log n)$ temporary storage is still higher than the bound presented in this paper, which is reduced to 0 by partly overwriting the input polynomial and restoring it before the end. With this advantage we can modify the algorithms presented by Roche \cite{roche2009space} for the quantum setting. The algorithms in this paper have an exponential speedup over other quantum algorithms that do not use extra space \cite{parent2017karatsuba}. Other variants that reach the same speedup as classical Karatsuba multiplication in the quantum setting so far have have done so at the cost of space \cite{kepley2015quantum}.

\subsection{Overview}
We introduce our notation for quantum computing by giving the elementary quantum gates in section \ref{s:background}. Our new multiplication algorithm needs several subroutines, specifically modular shifts and multiplication by a constant polynomial, introduced in section \ref{s:basic}. We introduce a Quantum Karatsuba algorithm for multiplication without reduction in section \ref{s:mult} and in binary finite fields in section \ref{s:revkar}. Both algorithms run without ancillary qubits and have a sub-quadratic Toffoli gate count. We implemented the algorithm in a simulated quantum computer and present the gate counts for specific finite fields in section \ref{s:results}.

\section{Quantum background}\label{s:background}
Quantum computing uses reversible gates, which unlike classical gates can be run in reverse and require an equal number of input and output quantum bits (qubits). In this paper we will not make use of the quantum properties of qubits, but the gates we use can be applied to superpositions of qubits in states 1 and 0. For the purpose of multiplication we need two gates to do reversible addition and multiplication:
\begin{itemize}
    \item The CNOT, or Feynman, gate serves as the equivalent of XOR or $\mathbb{F}_2$-addition. This gate takes 2 qubits as inputs and adds one input to the other qubit and outputs the other qubit as itself: $(a,b)\rightarrow (a\oplus b,b)$. It is reversible and its own inverse: applying it twice would result in $(a \oplus b\oplus b,b)=(a,b)$. In Circuit \ref{c:CNOT} an example has been drawn. In algorithms we write this as $a\leftarrow\text{CNOT}(a,b)$.
    \item The Toffoli (TOF) gate serves as the equivalent of AND or $\mathbb{F}_2$-multiplication in our case. This gate takes 3 qubits as inputs and adds the result of mulitplication of the frist two qubits to the third qubit and outputs the other qubits as themselves: $(a,b,c)\rightarrow(a,b,c\oplus (a\cdot b))$. It is also its own inverse. In circuit \ref{c:TOF} an example has been drawn. In algorithms we write this as $c\leftarrow\text{TOF}(a,b,c)$
\end{itemize}
\begin{figure}[h]\centering\makebox{\Qcircuit @C=1em @R=0.7em @! {
a& & \targ &\qw & a\oplus b \\
b& & \ctrl{-1} & \qw & b}}
\caption{The CNOT gate}
\label{c:CNOT}
\end{figure}
\begin{figure}[h]\centering\makebox{\Qcircuit @C=1em @R=0.7em @! {
a& & \ctrl{2} &\qw && a \\
b& & \ctrl{1} &\qw && b \\
c& & \targ & \qw && c\oplus (a\cdot b)}}
\caption{The TOF gate}
\label{c:TOF}
\end{figure}
In addition to these operations, we will also need to swap some qubits. Unlike the previous gates we do not build these in physical circuits. Rather, we change the index on some qubits: if we were to swap qubits 1 and 2 we would simply refer to qubit 1 as ``2'' and qubit 2 as ``1'' from that point on without counting any quantum gates. In Circuit \ref{c:swap} an example has been drawn.
\begin{figure}[h]\centering\makebox{\Qcircuit @C=1em @R=0.7em @! {
a& & \qswap &\qw & b \\
b& & \qswap\qwx & \qw & a}}
\caption{The swap}
\label{c:swap}
\end{figure}
\\These 3 actions are the only essential ones we use in this paper. Although none of these are explicit quantum actions, the quantum dimension comes from optimizing for low Toffoli gate count. Currently no large quantum computer exists but current estimates put the cost of one Toffoli gate at many times that of a CNOT gate.

\section{Basic Arithmetic}\label{s:basic}
In this section we discuss reversible in-place algorithms for the basic arithmetic of binary polynomials.

\subsection{Addition and binary shift}\label{s:shift}

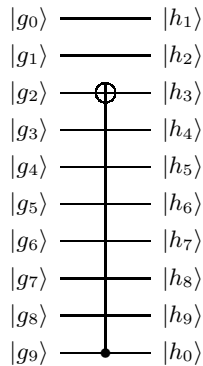
\begin{figure}\centering\makebox{\Qcircuit @C=1em @R=0.7em @! {
\lstick{\ket{g_0}} & \qw & \rstick{\ket{h_1}} \qw \\
\lstick{\ket{g_1}} & \qw & \rstick{\ket{h_2}} \qw \\
\lstick{\ket{g_2}} & \targ & \rstick{\ket{h_3}} \qw \\
\lstick{\ket{g_3}} & \qw & \rstick{\ket{h_4}} \qw \\
\lstick{\ket{g_4}} & \qw & \rstick{\ket{h_5}} \qw \\
\lstick{\ket{g_5}} & \qw & \rstick{\ket{h_6}} \qw \\
\lstick{\ket{g_6}} & \qw & \rstick{\ket{h_7}} \qw \\
\lstick{\ket{g_7}} & \qw & \rstick{\ket{h_8}} \qw \\
\lstick{\ket{g_8}} & \qw & \rstick{\ket{h_9}} \qw \\
\lstick{\ket{g_9}} & \ctrl{-7} & \rstick{\ket{h_0}} \qw}}
\caption{Binary shift circuit for $\mathbb{F}_{2^{10}}$ with $g_0+\cdots+g_9x^9$ as the input and $h_0+\cdots+h_9x^9=g_9+g_0x+g_1x^2+(g_2+g_9)x^3+g_3x^4+\cdots+g_9x^9$ as the output.}\label{c:modshift}\end{figure}

\noindent The first operation we consider, addition, can easily be implemented for binary polynomials. Individual additions can be done with a CNOT gate, the addition of two polynomials of degree at most $n$ takes $n+1$ CNOT gates with depth 1. This operation uses ancillary qubits and the result of the addition replaces either of the inputs. Since addition is component-wise, addition for polynomials over $\mathbb{F}_2$ is the same as  addition for elements of the field $\mathbb{F}_{2^n}$.

Binary shifts are straightforward: they correspond to multiplying or dividing by $x$. This requires no quantum computation by doing a series of swaps.

Finally, if we have a fixed $n$, a polynomial $g(x)$ of degree at most $n-1$ and want to do a multiplication by $x$ followed by a modular reduction by a fixed weight-$\omega$ and degree-$n$ polynomial $m(x)$ that has coefficient 1 for $x^0$, we can do this in 2 steps. We represent $m(x)$ as $M$ where $M$ is an ordered list of length $\omega$ that contains the degrees of the nonzero terms in descending order, for example if $m(x)=1+x^3+x^{10}$ we get $M=[10,3,0]$. Let $g(x)=\sum_{i=0}^{n-1}g_ix^i$:
\begin{itemize}
    \item Step 1: For every qubit $g_i$ change its index so that it represents the coefficient of $x^{i+1\text{ mod }n}$. Let $h_i$ be the coefficients of the relabeled polynomial, i.e. $h_{i+1\text{ mod }n}=g_i$.
    \item Step 2: Apply CNOT controlled by the $x^0$ term $h_0$ ($g_{n-1}$ before Step 1) to $h_j$, with $j=M_1,\ldots,M_{\omega-2}$. In the example of $1+x^3+x^{10}$ we would apply 1 CNOT to $h_3$ controlled by $h_0$.
\end{itemize}

\noindent See Circuit \ref{c:modshift} for an example. After a multiplication by $x$ the coefficient of $x^0$ is always 0. Since $m(x)$ always has coefficient 1 for $x^0$, after a reduction by $m(x)$ that qubit will be 1 and if no reduction takes place that qubit is 0, which means our modular shift algorithm is always reversible. This results in a total of $\omega-2$ CNOT gates for a modular reduction, with depth $\omega-2$ and we do not use ancillary qubits. Since we use reversible gates, running this circuit in reverse corresponds to dividing by $x$ modulo $m(x)$.

\subsection{Multiplication by a constant polynomial}
\begin{figure}[h]\centering\makebox{\Qcircuit @C=1em @R=0.7em @! {
\lstick{\ket{g_0}} & \targ & \qw & \qw & \ctrl{3} & \qw & \qw & \rstick{\ket{g_0+g_2}} \qw\\
\lstick{\ket{g_1}} & \qw & \targ & \targ & \qw & \ctrl{1} & \qw & \rstick{\ket{g_1+g_2+g_3}} \qw\\
\lstick{\ket{g_2}} & \ctrl{-2} & \ctrl{-1} & \qw & \qw & \targ & \qswap & \rstick{\ket{g_0+g_2+g_3}} \qw\\
\lstick{\ket{g_3}} & \qw & \qw & \ctrl{-2} & \targ & \qw & \qswap \qwx & \rstick{\ket{g_1+g_3}} \qw
}}\caption{Multiplication of $g$ by $1+x^2$ modulo $1+x+x^4$. Depth 4 and 5 CNOT gates.}\label{LUPmult}\end{figure}
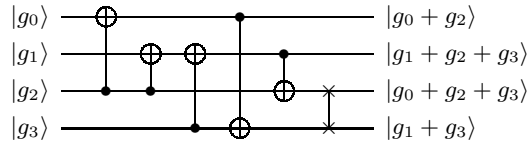

\noindent Multiplication by a constant non-zero polynomial in a fixed binary field is $\mathbb{F}_2$-linear: as the field polynomial is irreducible, every input corresponds to exactly one output. We can see that any such multiplication can be represented as a matrix, which we can turn into a circuit using an \textit{LUP}-decomposition, an algorithm also used by Amento, R\"otteler and Steinwandt \cite{amento2012efficient}. For example, multiplication by $1+x^2$ modulo $1+x+x^4$ can be represented by a matrix $\Gamma$. Using the decomposition $\Gamma=P^{-1}LU$ we get an upper and lower triangular matrix which we can translate into a circuit. Any 1 not on the diagonal in $U$ and $L$ is a CNOT controlled by the column number on the row number. In cases of conflict\footnote{Conflicts exist if according to the triangular matrix a CNOT would both have to be applied on and controlled by a qubit. By doing the controlled operation first and applying the operation on it afterwards, we ensure that the matrix multiplication is correctly translated.}, for $U$ CNOT gates should be performed top row first, second row second and so on and for $L$ CNOT gates from the bottom row up. $P$ represents a series of swaps, and can be represented either as a permutation matrix or an ordered list with all elements from $0$ to $n-1$.
$$\Gamma=\left(\begin{matrix}
1 & 0 & 1 & 0\\
0 & 1 & 1 & 1\\
1 & 0 & 1 & 1\\
0 & 1 & 0 & 1
\end{matrix}\right)=P^{-1}LU=\left(\begin{matrix}
1 & 0 & 0 & 0\\
0 & 1 & 0 & 0\\
0 & 0 & 0 & 1\\
0 & 0 & 1 & 0
\end{matrix}\right)\left(\begin{matrix}
1 & 0 & 0 & 0\\
0 & 1 & 0 & 0\\
0 & 1 & 1 & 0\\
1 & 0 & 0 & 1
\end{matrix}\right)\left(\begin{matrix}
1 & 0 & 1 & 0\\
0 & 1 & 1 & 1\\
0 & 0 & 1 & 0\\
0 & 0 & 0 & 1
\end{matrix}\right)$$
Circuit \ref{LUPmult} shows how we translate $\Gamma$. According to \cite{amento2012efficient} this costs up to $n^2 + n$ CNOT gates with depth up to $2n$. We can improve this count by noting $L$ and $U$ are each size $n$ by $n$ and can have up to $(n^2-n)/2$ non-diagonal non-zero entries, giving us up to $n^2-n$ CNOT gates. Note that the \textit{LUP}-decomposition is precomputed and for any fixed polynomial and field we can give an exact CNOT gate count and depth.

Since this algorithm is introduced in \cite{amento2012efficient} without correctness proof and we will use it later for a bigger algorithm, we will write an explicit implementation and go over the correctness of this algorithm. Note that since we are working with reversible algorithms, multiplying by constant $f(x)$ is the same as doing the reverse of multiplying by constant $f(x)^{-1}$.

\begin{algorithm}[h]
\DontPrintSemicolon
\SetKwInOut{Input}{Fixed input}\SetKwInOut{Output}{Quantum Input}
\Input{A binary \textit{LUP}-decomposition $L,U,P^{-1}$ for a binary $n$ by $n$ matrix that corresponds to multiplication by the constant polynomial $f(x)$ in the field $\mathbb{F}_2[x]/m(x)$.}
\SetKwInOut{Input}{Quantum input}
\Input{A binary polynomial $g(x)$ of degree up to $n-1$ stored in an array $G$.}
\KwResult{$G$ as $f\cdot g$ in the field $\mathbb{F}_2/m(x)$.}
\For{$i=0..n-1$\tcp*{$U\cdot G$}}{
\For{$j=i+1..n-1$}{
\If{$U[i,j]=1$}{
$G[i]\leftarrow\text{CNOT}(G[i],G[j])$
}}}
\For{$i=n-1..0$\tcp*{$L\cdot UG$}}{
\For{$j=i-1..0$}{
\If{$L[i,j]=1$}{
$G[i]\leftarrow\text{CNOT}(G[i],G[j])$
}}}
\For{$i=0..n$\tcp*{$P^{-1}\cdot LUG$}}{
\For{$j=i+1..n-1$}{
\If{$P^{-1}[i,j]=1$}{
SWAP$(G[i],G[j])$\\
SWAP column $i$ and $j$ of $P^{-1}$
}}}
\caption{MULT$_{f(x)}$, from \cite{amento2012efficient}. Reversible algorithm for in-place multiplication by a nonzero constant polynomial $f(x)$ in $\mathbb{F}_2[x]/m(x)$ with $m(x)$ an irreducible polynomial. \label{a:constmult}}
\end{algorithm}

\begin{theorem}
Algorithm \ref{a:constmult} correctly describes multiplication by a non-zero constant polynomial in a fixed binary field.
\end{theorem}

\begin{proof}
Since multiplication by a non-zero constant in a finite field is a linear map, an invertible matrix $\Gamma$ to represent this linear map must exist. Since $\Gamma$ is invertible, its decomposition $L,U,P^{-1}$ must also consist of linear maps. Since we are working in a binary field and $U$ is an invertible upper-triangular matrix, the diagonal of $U$ is all-one. If we look at lines 1 through 4 of the algorithm, we can see it corresponds to applying linear map $U$ to $g$, as it results in $g_i=\sum_{j=0}^{n-1}u_{i,j}g_j$ for $i=0,..,n-1$. Analogously the same is true for $L$ in lines 5 through 8. We can also see that if $P^{-1}$ is a row-permutation of the identity matrix, lines 9 through 13 will apply it correctly. Since $P^{-1}LU=\Gamma$ we have correctly applied the linear map $\Gamma$.
\end{proof}

\noindent Note that the algorithm is not optimized for depth, for example in circuit \ref{LUPmult} the first and second CNOT could be swapped so the depth would be 3 rather than 4.

\subsubsection{Choice of field polynomials}\hspace*{\fill}\\
When doing operations in a finite binary field we can choose what representation we use, as long as the polynomial $m(x)$ is irreducible. Our goal is to make the matrices $L$ and $U$ as sparse as possible. For this purpose we also want our $\Gamma$ to be as sparse as possible, which can be achieved in two steps: pick irreducible polynomials with as few non-zero coefficients as possible, i.e. trinomials when available and pentanomials otherwise, and pick irreducible polynomials where the second highest non-constant term has the lowest possible degree. For example, the pentanomial $1+x^3+x^4+x^{19}+x^{20}$ would require 108 CNOT gates, the pentanomial $1+x^3+x^5+x^9+x^{20}$ would require 55 CNOT gates, while the trinomial $1+x^3+x^{20}$ would require only 27. All 3 polynomials are irreducible. 
In Table \ref{t:constmult} we can see some examples of gate counts for various choices of $n$. The depth count is an upper bound without accounting for swapping gates.

\begin{table}[h]
\begin{center}
\begin{tabular}{|c|c|c|c|c|}
    \hline
    Degree & Irreducible polynomial & Source & CNOT gates & Depth upper bound \\
    \hline
    4 & $[4,1,0]$ &\cite{2005ehcc} & 5 & 4\\
    8 & $[8,4,3,1,0]$ &\cite{2005ehcc} & 20 & 14\\
    16 & $[16,5,3,1,0]$ &\cite{2005ehcc} & 47 & 30\\
    32 & $[32,7,3,2,0]$ &\cite{2005ehcc} & 133 & 93\\
    64 & $[64,4,3,1,0]$ &\cite{2005ehcc} & 264 & 182\\
    127 & $[127,1,0]$ &\cite{2005ehcc} & 396 & 293 \\
    128 & $[128,7,2,1,0]$ & \cite{2005ehcc} & 626 & 443\\
    163 & $[163,7,6,3,0]$ &\cite{fips2013186}& 740 & 975 \\
    163 & $[163,89,74,15,0]$ &\cite{banegas2018new} & 1885 & 1646 \\ 
    233 & $[233,74,0]$ & \cite{fips2013186} & 3319 & 2976 \\
    256 & $[256,10,5,2,0]$ & \cite{2005ehcc} & 1401 & 1030 \\
    283 & $[283,12,7,5,0]$ &\cite{fips2013186} & 2117 & 1700 \\
    283 & $[283,160,123,37,0]$ &\cite{banegas2018new} & 6785 & 6368\\
    571 & $[571,10,5,2,0]$ &\cite{fips2013186} & 4027 & 3177 \\
    571 & $[571,353,218,135,0]$ &\cite{banegas2018new} & 33182 & 32331 \\
    1024 & $[1024, 19,6,1, 0]$ &\cite{seroussi1998table} & 8147 & 6624\\
    \hline
\end{tabular}
\end{center}
\caption{Comparison of the CNOT gates required for various instances of Algorithm \ref{a:constmult}. Source is the source of the polynomial.}\label{t:constmult}
\end{table}

\newpage
\section{Quantum Multiplication for binary polynomials}\label{s:mult}
This section details schoolbook multiplication and we present our new Karatsuba algorithm. 
\subsection{Quantum Schoolbook Multiplication}
The simplest way to multiply is schoolbook multiplication. For two polynomials of degree at most $n-1$ that takes $n^2$ Toffoli gates, the number of pairs of qubits from the first and second polynomial. While the computation does not use ancillary qubits, the result needs to be stored separately from input in $2n-1$ qubits; unlike the previous circuits we cannot replace either of the inputs with the result since the Toffoli gate requires a separate output. If we want to apply modular reduction steps by a weight-$k$ and degree-$n$ odd polynomial, this adds $(n-1)\cdot(k-2)$ CNOT gates and uses no ancillary qubits (by using the modular shift algorithm after every $n$ multiplications). The result is stored in $n$ qubits.

\subsection{Classic Karatsuba multiplication in binary polynomial rings}\label{s:ck}
Rather than using schoolbook multiplication, methods like Karatsuba multiplication \cite{karatsuba1962multiplication} can speed up multiplication of large numbers. We can look at in-place multiplication in the classical case for ideas \cite{roche2009space}. As input we take two polynomials of size up to $n$, $f(x)$ and $g(x)$ as well as a polynomial of size $2n$: $h(x)$. As output we desire $h+f\cdot g$. For some $k$ such that $\frac{n}{2}\leq k<n$ (we will always use $k=\lceil \frac{n}{2}\rceil$) we can split each polynomial as follows: $f=f_0+f_1x^k$, $g=g_0+g_1x^k$ and $h=h_0+h_1x^k+h_2x^{2k}+h_3x^{3k}$.

We compute intermediate products $\alpha=f_0\cdot g_0$, $\beta=f_1\cdot g_1$ and $\gamma=(f_0+f_1)\cdot(g_0+g_1)$. Finally, we add these in the right way for Karatsuba multiplication: $$h+f\cdot g=h+\alpha+(\gamma+\alpha+\beta)x^k+\beta x^{2k}.$$ For cleanliness, we can split up our $\alpha, \beta, \gamma$ in the same way as $f$ and $g$ to get a result with no overlap, which is useful for checking correctness:
$$h+f\cdot g=(h_0+\alpha_0)+(h_1+\alpha_0+\alpha_1+\beta_0+\gamma_0)x^k+(h_2+\alpha_1+\beta_0+\beta_1+\gamma_1)x^{2k}+(h_3+\beta_1)x^{3k}.$$
Alternatively, we can rewrite this another way that will prove useful:
$$h+f\cdot g=h+(1+x^k)\alpha+x^k\gamma+x^k(1+x^k)\beta.$$

\subsection{Reversible Karatsuba multiplication in binary polynomial rings}
Based on these equations we can split our multiplication algorithm into 2 parts: given $f(x),g(x),h(x)$ calculate $h+f\cdot g$ and given $k,f(x),g(x),h(x)$ with $k>\max(\text{deg}(f),\text{deg}(g))$ calculate $h+(1+x^k)f\cdot g$. We will look at our algorithms for the 2 parts, which can then be used recursively to provide a significant improvement to the schoolbook algorithm in terms of Toffoli gate count.

\begin{algorithm}[h]
\DontPrintSemicolon
\SetKwInOut{Input}{Fixed input}\SetKwInOut{Output}{Quantum Input}
\Input{A constant integer $k>0$ to indicate part size as well as an integer $n\leq k$ to indicate polynomial size. $\ell=\max(0,2n-1-k)$ is the size of $h_2$ and $(fg)_1$. In the case of Karatsuba we will have either $n=k$ or $n=k-1$.}
\SetKwInOut{Input}{Quantum input}
\Input{Two binary polynomials $f(x),g(x)$ of degree up to $n-1$ stored in arrays $A$ and $B$ respectively of size $n$. A binary polynomial $h(x)$ of degree up to $k+2n-2$ stored in array $C$ of size $2k+\ell$.}
\KwResult{$A$ and $B$ as input, $C$ as $h+(1+x^k)fg$}
\If{$n>1$}{
$C[k..k+\ell-1]\leftarrow \text{CNOT}(C[k..k+\ell-1],C[2k..2k+\ell-1])$\\
$C[0..k-1]\leftarrow \text{CNOT}(C[0..k-1],C[k..2k-1])$\\
$C[k..2k+\ell-1]\leftarrow \text{KMULT}(A[0..n-1],B[0..n-1],C[k..2k+\ell-1])$\\
$C[0..k-1]\leftarrow \text{CNOT}(C[0..k-1],C[k..2k-1])$\\
$C[k..k+\ell-1]\leftarrow \text{CNOT}(C[k..k+\ell-1],C[2k..2k+\ell-1])$
}\Else{
$C[0]\leftarrow\text{CNOT}(C[0],C[k])$\\
$C[k]\leftarrow\text{TOF}(A[0],B[0],C[k])$\\
$C[0]\leftarrow\text{CNOT}(C[0],C[k])$}
\caption{MULT1x$_k$. Reversible algorithm for multiplication by the polynomial $1+x^k$. \label{a:1xk}}
\end{algorithm}

\begin{table}[h]
    \centering
    \begin{tabular}{|c|c|c|c|}
        \hline
        Line &\multicolumn{3}{c|}{$C$ in MULT1x$_k$ }\\ \cline{2-4}
        & $C[0..k-1]$ & $C[k..2k-1]$ & $C[2k..2k+\ell-1]$ \\
        \hline
        1 & $h_0$&$h_1$&$h_2$ \\
        2 & $h_0$&$h_1+h_2$&$h_2$ \\
        3 & $h_0+h_1+h_2$&$h_1+h_2$&$h_2$ \\
        4 & $h_0+h_1+h_2$&$h_1+h_2+(fg)_0$&$h_2+(fg)_1$ \\
        5 & $h_0+(fg)_0$&$h_1+h_2+(fg)_0$&$h_2+(fg)_1$ \\
        6 & $h_0+(fg)_0$&$h_1+(fg)_0+(fg)_1$&$h_2+(fg)_1$ 
        \\ \hline
    \end{tabular}
    \caption{Step by step calculation of Algorithm \ref{a:1xk}.}
    \label{t:mult1xk}
\end{table}

\begin{lemma}\label{l:1xk}
Given polynomials $f,g$ of degree up to $n-1$ with $n>1$, polynomial $h$ of degree up to $k+2n-2$ with some $k\geq n$ and assuming Algorithm \ref{a:kmult} correctly calculates $h+fg$ with degrees of $f,g$ and $h$ bounded as above, Algorithm \ref{a:1xk} correctly calculates $h+(1+x^k)fg$ in $\mathbb{F}_2[x]$ without altering the values of $f$ and $g$.
\end{lemma}

\begin{proof}
Let $\ell=\max(0,2n-1-k)$. Table \ref{t:mult1xk} gives the result of each step on array C, split into 3 parts of size $k$, $k$ and $\ell-1$ respectively: $h=h_0+h_1x^k+h_2x^{2k}$. The final result corresponds to $h_0+(fg)_0+(h_1+(fg)_0+(fg)_1)x^k+(h_2+(fg)_1)x^{2k}=h_0+h_1x^k+h_2x^{2k}+fg+fgx^k=h+(1+x^k)fg$, where $(fg)_0$ is the first $k$ terms of $f\cdot g$ and $(fg)_1$ is the last up to $\ell$ terms.\\\indent
$f$ and $g$ do not have their values altered because arrays $A$ and $B$ remain unchanged. 
\end{proof}

\noindent Algorithm \ref{a:1xk} computes $h+(1+x^k)fg$ with at most $2k+2\ell\geq 2k+2(2n-1-k)= 4n-2$ CNOT gates, at a depth of 4 per layer and 1 call to Algorithm \ref{a:kmult} for an $n$-by-$n$ multiplication. For $n=1$ both the depth and number of gates is 2 CNOT and 1 TOF gates.

\begin{algorithm}[h]
\DontPrintSemicolon
\SetKwInOut{Input}{Fixed input}\SetKwInOut{Output}{Quantum Input}
\Input{A constant integer $n$ to indicate polynomial size and an integer $k< n\leq 2k$ with $k=\lceil\frac{n}{2}\rceil$ for $n>1$ and $k=0$ for $n=1$, to indicate upper and lower half.}
\SetKwInOut{Input}{Quantum input}
\Input{Two binary polynomial $f,g$ of degree up to $n-1$ stored in arrays $A$ and $B$ respectively of size $n$. A binary polynomial $h$ of degree up to $2n-2$ stored in array $C$ of size $2n-1$.}
\KwResult{$A$ and $B$ as input, $C$ as $h+fg$}
\If{$n>1$}{
$C[0..3k-2]\leftarrow \text{MULT1x}_k(A[0..k-1],B[0..k-1],C[0..3k-2])$\\
$C[k..2n-2]\leftarrow \text{MULT1x}_k(A[k..n-1],B[k..n-1],C[k..2n-2])$\\
$A[0..n-k-1]\leftarrow \text{CNOT}(A[0..n-k-1],A[k..n-1])$\\
$B[0..n-k-1]\leftarrow \text{CNOT}(B[0..n-k-1],B[k..n-1])$\\
$C[k..3k-2]\leftarrow \text{KMULT}(A[0..k-1],B[0..k-1],C[k..3k-2])$\\
$B[0..n-k-1]\leftarrow \text{CNOT}(B[0..n-k-1],B[k..n-1])$\\
$A[0..n-k-1]\leftarrow \text{CNOT}(A[0..n-k-1],A[k..n-1])$
}\Else{$C[0]\leftarrow\text{TOF}(A[0],B[0],C[0])$}
\caption{KMULT. Reversible algorithm for multiplication of 2 polynomials. \label{a:kmult}}
\end{algorithm}

\begin{lemma}\label{l:kmult}
Let $k=\lceil\frac{n}{2}\rceil$. Given polynomials $f,g$ of degree up to $n-1$ with $n>1$ and $h$ of degree up to $2n-2$. Assuming Algorithm \ref{a:1xk} correctly calculates $h'+(1+x^k)f'g'$ for $f',g'$ up to degree $k-1$ and $h'$ up to degree $3k-2$, and Algorithm \ref{a:kmult} correctly calculates $h''+f''g''$ with $f'',g''$ of degree $k-1$ and $h''$ of degree $2k-2$ without altering the values of $f''$ and $g''$. Then Algorithm \ref{a:kmult} correctly calculates $h+fg$ in $\mathbb{F}_2[x]$. The values of $f$ and $g$ are the same after the algorithm as they were before.
\end{lemma}

\begin{proof}
Table \ref{t:infmult} gives the result of each line on array $C$, split into 4 parts of size $k$, $k$, $k$ and $2n-1-3k$ respectively: $h=h_0+h_1x^k+h_2x^{2k}+h_3x^{3k}$. As can be seen in the table, the final result corresponds to $(h_0+\alpha_0)+(h_1+\alpha_0+\alpha_1+\beta_0+\gamma_0)x^k+(h_2+\alpha_1+\beta_0+\beta_1+\gamma_1)x^{2k}+(h_3+\beta_1)x^{3k}=h+f\cdot g$ as discussed in Section \ref{s:ck}. Lines 7 and 8 are the inverses of lines 4 and 5 so return $A$ and $B$ to their original states.
\end{proof}

\begin{table}[h]
    \centering
    \begin{tabular}{|c|c|c|c|c|}
        \hline
        Line&\multicolumn{4}{c|}{$C$ in KMULT} \\\cline{2-5}
        &$C[0..k-1]$&$C[k..2k-1]$&$C[2k..3k-1]$&$C[3k..2n-2]$\\
        \hline
        1 & $h_0$&$h_1$&$h_2$&$h_3$\\
        2 & $h_0+\alpha_0$&$h_1+\alpha_0+\alpha_1$&$h_2+\alpha_1$&$h_3$\\
        3-5 & $h_0+\alpha_0$&$h_1+\alpha_0+\alpha_1+\beta_0$&$h_2+\alpha_1+\beta_0+\beta_1$&$h_3+\beta_1$\\
        6-8 & $h_0+\alpha_0$&$h_1+\alpha_0+\alpha_1+\beta_0+\gamma_0$&$h_2+\alpha_1+\beta_0+\beta_1+\gamma_1$&$h_3+\beta_1$
        
        \\ \hline
    \end{tabular}
    \caption{Step by step calculation of Algorithm \ref{a:kmult}.}
    \label{t:infmult}
\end{table}

\noindent Algorithm \ref{a:kmult} computes $h+fg$ with $4(n-k)$ CNOT gates, at a depth of 4, 1 call to itself for a $k$-by-$k$ multiplication, 1 call to Algorithm \ref{a:1xk} for a $k$-by-$k$ multiplication and 1 call to Algorithm \ref{a:1xk} for an $(n-k)$-by-$(n-k)$ multiplication. For $n=1$ we just have a single TOF gate.

\begin{theorem}\label{th:kmult}
Given polynomials $f,g$ of degree up to $n-1$ and $h$ of degree up to $2n-2$, Algorithm \ref{a:kmult} correctly calculates $h+fg$. The values of $f$ and $g$ are the same after the algorithm as they were before.
\end{theorem}

\begin{proof}
We use proof by induction. For $n=1$ line 10 of Algorithm \ref{a:kmult} correctly calculates $h+fg$ without altering $f$ or $g$.\\\indent
For $n=2$ two calls are made to Algorithm \ref{a:1xk} and one call to Algorithm \ref{a:kmult} with $n'=1$ and $k'=1$. Lines 7-9 of Algorithm \ref{a:1xk} correctly calculate $h'+(1+x^k)f'g'$.\\\indent
For $n>2$ we use lemmas \ref{l:1xk} and \ref{l:kmult} as our inductive steps. Every time Algorithm \ref{a:kmult} is called recursively to calculate $h'+f'g'$ with $f',g'$ of degree $n'-1$, it is with either $n'=\lceil\frac{n}{2}\rceil$ or $n'=n-\lceil\frac{n}{2}\rceil=\lfloor\frac{n}{2}\rfloor$.\\\indent 
The series $\lceil\frac{n}{2}\rceil,\lceil\frac{\lceil\frac{n}{2}\rceil}{2}\rceil,\lceil\frac{\lceil\frac{\lceil\frac{n}{2}\rceil}{2}\rceil}{2}\rceil,...$ reaches 1 in $O(\log n)$ steps and $\lfloor\frac{n}{2}\rfloor\leq\lceil\frac{n}{2}\rceil$. From this we can see that we reach $n'=1$ or 2 in finite steps. By induction Algorithm \ref{a:kmult} correctly calculates $h+fg$ and returns $f$ and $g$ to their original values.
\end{proof}

\section{Reversible Karatsuba multiplication in binary finite fields}\label{s:revkar}
With this basis, we can move on to the modular multiplication. We will need Algorithm \ref{a:constmult}, which we will also run in reverse for multiplication by an inverse, and the binary shifts from Section \ref{s:shift}, which we will refer to as MODSHIFT, as well as the previous Karatsuba algorithms. Unlike before, we will assume we start with an all-zero input. If we would want to add the multiplication result to a polynomial in C we would have to prepare it by first performing $k$ divisions by $x$ (reverse MODSHIFT) which would take $k$ or $3k$ CNOT gates for trinomials and pentanomials respectively. We can see in Algorithm \ref{a:modmult} the number of operations we use:
\begin{itemize}
    \item 3 calls to Algorithm \ref{a:kmult}: twice for $k$-by-$k$ multiplication and once for $(n-k)$-by-$(n-k)$ multiplication.
    \item 2 calls to Algorithm \ref{a:constmult} (once in reverse), each time for multiplication by the same polynomial $1+x^k$.
    \item $k$ calls to MODSHIFT.
    \item 4 times $(n-k)$ CNOT gates, half of which can be performed at the same time.
\end{itemize}
Note that Algorithm \ref{a:kmult} can multiply two polynomials $f$ and $g$ of degree at most $\lceil\frac{n}{2}\rceil-1$ while needing $n$ space for the output polynomial $h$, which has degree $n-1$ at most in the case that $n$ is odd. We make recursive calls to Algorithm \ref{a:kmult} rather than Algorithm \ref{a:modmult} because it uses significantly fewer CNOT operations and fits in the required space.

\begin{table}[h]
    \centering
    \begin{tabular}{|c|c|}
        \hline
        Line & $C$ in MODMULT \\
        \hline
        1,2 & 0 \\
        3-5 & $\gamma$\\
        6 & $(1+x^k)^{-1}\gamma\mod m$ \\
        7 & $(1+x^k)^{-1}\gamma+\beta\mod m$\\
        9 & $x^k((1+x^k)^{-1}\gamma+\beta)\mod m$\\
        10 & $\alpha+x^k((1+x^k)^{-1}\gamma+\beta)\mod m$\\
        11 & $(1+x^k)\alpha+x^k\gamma+x^k(1+x^k)\beta\mod m$
        \\ \hline
    \end{tabular}
    \caption{Step-by-step calculation of Algorithm \ref{a:modmult}.}
    \label{t:finmult}
\end{table}

\begin{algorithm}[h]
\DontPrintSemicolon
\SetKwInOut{Input}{Fixed input}\SetKwInOut{Output}{Quantum Input}
\Input{A constant integer $n$ to indicate field size, $k=\lceil\frac{n}{2}\rceil$. $m(x)$ of degree $n$ as the field polynomial. The \textit{LUP}-decomposition precomputed for multiplication by $1+x^k$ modulo $m(x)$.}
\SetKwInOut{Input}{Quantum input}
\Input{Two binary polynomials $f(x),g(x)$ of degree up to $n-1$ stored in arrays $A$ and $B$ respectively of size $n$. An all-zero array $C$ of size $n$}
\KwResult{$A$ and $B$ as input, $C$ as $f\cdot g\mod m$.}
$A[0..n-k-1]\leftarrow \text{CNOT}(A[0..n-k-1],A[k..n-1])$\\
$B[0..n-k-1]\leftarrow \text{CNOT}(B[0..n-k-1],B[k..n-1])$\\
$C[0..n-1]\leftarrow\text{KMULT}(A[0..k-1],B[0..k-1],C[0..n-1])$\\
$B[0..n-k-1]\leftarrow \text{CNOT}(B[0..n-k-1],B[k..n-1])$\\
$A[0..n-k-1]\leftarrow \text{CNOT}(A[0..n-k-1],A[k..n-1])$\\
$C[0..n-1]\leftarrow\text{MULT}^{-1}_{1+x^k}(C[0..n-1])$ \\
$C[0..n-1]\leftarrow\text{KMULT}(A[k..n-1],B[k..n-1],C[0..n-1])$\\
\For{$i=0..k-1$}{
$C[0..n-1]\leftarrow\text{MODSHIFT}(C[0..n-1])$}
$C[0..n-1]\leftarrow\text{KMULT}(A[0..k-1],B[0..k-1],C[0..n-1])$\\
$C[0..n-1]\leftarrow\text{MULT}_{1+x^k}(C[0..n-1])$
\caption{MODMULT. Reversible algorithm for multiplication of 2 polynomials in $\mathbb{F}_2[x]/m(x)$ with $m(x)$ an irreducible polynomial. \label{a:modmult}}
\end{algorithm}

\begin{theorem}
Algorithm \ref{a:modmult} correctly calculates $fg$ in a field $\mathbb{F}_2[x]/m(x)$ and the values of $f$ and $g$ are the same after the algorithm as they were before.
\end{theorem}
\begin{proof}
Table \ref{t:finmult} gives the result of each line on array $C$. As can be seen in the table, the final result corresponds to $(1+x^k)\alpha+x^k\gamma+ x^k(1+x^k)\beta\mod m$. Lines 4 and 5 are the inverses of lines 1 and 2 so return $A$ and $B$ to their original states.
\end{proof}

\section{Results}\label{s:results}
Algorithm \ref{a:modmult} uses the same number of Toffoli gates as regular Karatsuba multiplication: 3 half-sized multiplications. This means the asymptotic number of Toffoli gates is the same as for regular Karatsuba: $O(n^{\log(3)})\approx O(n^{1.58})$. This is a significant improvement over the $n^2$ Toffoli gates required for schoolbook multiplication. The number of CNOT gates is less clear as the number of CNOT gates required for the multiplications with constant polynomials is strongly dependent on our choice of field polynomial. It is not within the scope of this paper to find a stronger bound than $O(n^2)$ for the number of CNOT gates, which is currently used for the constant multiplication. In a strict comparison of these CNOT gates, this is worse than the $O(n)$ CNOT gates used by modular schoolbook multiplication, even if we can find a better estimate, but our primary goal is minimizing the number of Toffoli gates without introducing ancillary qubits. In our implementation, even the sum of CNOT and Toffoli gates ends up lower after some degree than the number of Toffoli gates for schoolbook multiplication.

\begin{table}[h]
    \centering
    \begin{tabular}{|c|c|ccc|}
        \hline
        Degree & schoolbook TOF gates & Algorithm \ref{a:modmult} TOF gates & CNOT gates & Depth upper bound \\
        \hline
        2 & 4 & 3 & 9 & 9 \\
        4 & 16 & 9 & 44 & 32 \\
        8 & 64 & 27 & 200 & 124 \\
        16 & 256 & 81 & 678 & 365\\
        32 & 1,024 & 243 & 2,238 & 1,110\\
        64 & 4,096 & 729 & 6,896 & 3,129\\
        127 & 16,129 & 2,185 & 20,632 & 8,769\\
        128 & 16,384 & 2,187 & 21,272 & 9,142\\
        163 & 26,569 & 4,387 & 37,168 & 17,906 \\
        233 & 54,289 & 6,323 & 63,655 & 29,530 \\
        256 & 65,536 & 6,561 & 64,706 & 26,725 \\
        283 & 80,089 & 10,273 & 89,620 & 41,548 \\ \hline
        571 & 326,041 & 31,171 & 270,940 & 121,821 \\
        1024 & 1,048,576 & 59,049 & 591,942 & 234,053
        \\ \hline
    \end{tabular}
    \caption{CNOT and TOF gate count and depth upper bounds for various instances of Algorithm \ref{a:modmult} as well as TOF gate count for schoolbook multiplication. Field polynomials used are the same as in Table \ref{t:constmult}, with the irreducible polynomial chosen that has the lowest CNOT count.}
    \label{t:time}
\end{table}
We implemented Algorithm \ref{a:modmult} in Java to simulate the execution. Code can be found in \cite{Git}. We used the program to automatically count the number of gates and give an estimate of the depth, see Table \ref{t:time} for the results. Depth count is done by maintaining a set of gates and checking every gate: if they overlap with the previous gate(s) the depth is increased by 1 and if they are not overlapping the gate is added to the set of gates to check against. The set of gates is cleared and replaced with the last gate whenever the depth is increased. The author is aware of methods to improve the depth but leaves this to future work.

When doing classical Karatsuba multiplication, the recursive Karatsuba multiplication is often substituted for schoolbook multiplication starting at a cutoff. For example, if multiplication is at most 7 times as expensive as addition, multiplication of two polynomials of degree at most 2 might be replaced by schoolbook multiplication to get 4 TOF gates instead of 3 TOF and 8 CNOT gates. However, the author is unaware of any realistic estimates of cost difference between CNOT and Toffoli gates where the difference is this small.

\subsection{Comparison to other instances of binary finite field multiplication}

\begin{table}[h]
    \centering
    \begin{tabular}{|c|ccc|ccc|ccc|}
    \hline
        Field size $2^n$ & \multicolumn{3}{c|}{Toffoli gates} & \multicolumn{3}{c|}{CNOT gates} & \multicolumn{3}{c|}{qubits} \\
        $n=$ & Here & \cite{kepley2015quantum} & \cite{maslov2009m} & Here & \cite{kepley2015quantum} & \cite{maslov2009m} & Here & \cite{kepley2015quantum} & \cite{maslov2009m}\\ \hline
        4 & 9 & 9 & 16 & 44 & 22 & 3 & 12 & 17 & 12\\
        16 & 81 & 81 & 256 & 678 & 376 & 45 & 48 & 113 & 48 \\
        127 & 2185 & 2185 & 16129 & 20632 & 13046 & 126 & 381 & 2433 & 381\\
        256 & 6561 & 6561 & 65536 & 64706 & 57008 & 765 & 768 & 7073 & 768\\
        $n$ & $O(n^{\log_2 3})$ & $O(n^{\log_2 3})$ & $n^2$ & $O(n^2)$ & $O(n^{\log_2 3})$ & $O(n)$ & $3n$ & $O(n^{\log_2 3})$ & $3n$ \\
        \hline
    \end{tabular}
    \caption{Comparison of this work with the works of Kepley and Steinwandt \cite{kepley2015quantum} and Maslov et al. \cite{maslov2009m} in terms of Toffoli and CNOT gates as well as qubit count.}
    \label{t:comparison}
\end{table}

\noindent We compare our algorithm to two previous instances of multiplication: a variant by Kepley and Steinwandt \cite{kepley2015quantum} that optimizes TOF gate count and a variant by Maslov, Mathew, Cheung and Pradhan \cite{maslov2009m} that does not use Karatsuba. Other variants exist, such as a Karatsuba variant by Parent, Roetteler and Mosca \cite{parent2017karatsuba}, that are worse in terms of space or Toffoli gate count. Since Kepley and Steinwandt use Clifford and T-gates rather than CNOT and Toffoli, we translate 7 of their T-gates and 8 Clifford gates to 1 Toffoli gate, and translate any remaining Clifford gates to CNOT. The resulting comparison is in Table \ref{t:comparison}. We can see that although Algorithm \ref{a:modmult} does not compare favorably in every regard, both the number of Toffoli gates and the number of qubits are best compared to the alternatives.

\section{Conclusion}
Algorithm \ref{a:modmult} provides a multiplication algorithm for binary polynomials in finite fields without using ancillary qubits and which has sub-quadratic Toffoli gate count. The CNOT gate count is high and the depth is not optimized, which is left open for future work: multiplication by a constant polynomial in $\mathbb{F}_{2^n}$ can likely be done in approximately linear time, which would bring down the theoretical CNOT gate count to the same order as classical Karatsuba. The saving in Toffoli gate count is the same as for Karatsuba on classical computers: for cryptographic field sizes the savings in Toffoli gates ranges from 80 to over 90 percent. This provides a basis for future work on elliptic curve problems on quantum computers as well as potential other work.

\subsubsection*{Acknowledgements} The author thanks Tanja Lange for her insights into quantum algorithms and classical finite field operations, Tanja Lange and Gustavo Banegas for their advice and supervision both on this paper and the master thesis this paper originates from, and to Daniel J. Bernstein for his insights into both quantum computing and classical multiplication algorithms.

\bibliographystyle{siam}
\bibliography{ms}

\end{document}